\documentclass{IEEEtran4PSCC}
\ifCLASSINFOpdf
   \usepackage[pdftex]{graphicx}
\else
   \usepackage[dvips]{graphicx}
\fi
\usepackage[cmex10]{amsmath}

\usepackage{cite}
\usepackage{amsthm}

\usepackage{amsmath,amssymb,amsfonts}
\usepackage{algorithmic}
\usepackage[ruled, lined, linesnumbered, commentsnumbered]{algorithm2e}
\usepackage{graphicx}
\usepackage{textcomp}
\usepackage{xcolor}
\usepackage{xspace}
\usepackage{nomencl}
\usepackage{glossaries} 
\usepackage{tikz}
\usetikzlibrary{calc}
\usepackage{ifthen}
\usepackage{siunitx}
\usepackage{booktabs}
\usepackage{pgfplots}
\usepackage{url}
\usepackage{multicol}
\usepackage{booktabs} 
\makeatletter
\let\old@ps@headings\ps@headings
\let\old@ps@IEEEtitlepagestyle\ps@IEEEtitlepagestyle
\def\psccfooter#1{%
    \def\ps@headings{%
        \old@ps@headings%
        \def\@oddfoot{\strut\hfill#1\hfill\strut}%
        \def\@evenfoot{\strut\hfill#1\hfill\strut}%
    }%
    \def\ps@IEEEtitlepagestyle{%
        \old@ps@IEEEtitlepagestyle%
        \def\@oddfoot{\strut\hfill#1\hfill\strut}%
        \def\@evenfoot{\strut\hfill#1\hfill\strut}%
    }%
    \ps@headings%
}
\makeatother

\psccfooter{%
        \parbox{\textwidth}{\hrulefill \\ \small{24th Power Systems Computation Conference} \hfill \begin{minipage}{0.2\textwidth}\centering \vspace*{4pt} \includegraphics[scale=0.06]{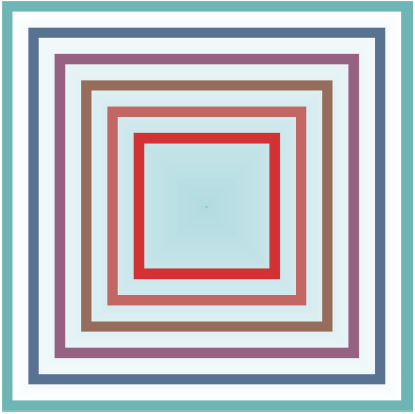}\\\small{PSCC 2026} \end{minipage} \hfill \small{Limassol, Cyprus --- June 8 -- June 12, 2026}}%
}

\let\originalleft\left
\let\originalright\right
\renewcommand{\left}{\mathopen{}\mathclose\bgroup\originalleft}
\renewcommand{\right}{\aftergroup\egroup\originalright}

  \newacronym[description={Markov-Chain Monte Carlo}]{MCMC}{MCMC}{Markov-Chain Monte Carlo Algorithm}

\newtheorem{proposition}{Proposition}

\begin{document}
%


\title{
Defending the power grid by segmenting the EV charging cyber infrastructure
}

\author{
\IEEEauthorblockN{Kirill Kuroptev, Florian Steinke}
\IEEEauthorblockA{Energy Information Networks {\&} Systems \\
Technical University of Darmstadt\\
Darmstadt, Germany\\
\{kirill.kuroptev,  florian.steinke\}@eins.tu-darmstadt.de}
\and
\IEEEauthorblockN{Efthymios Karangelos}
\IEEEauthorblockA{School of Electrical and Electronic Engineering \\
University College Dublin\\
Dublin, Ireland\\
efthymios.karangelos@ucd.ie}
}

\maketitle

\begin{abstract}

This paper examines defending the power grid against load-altering attacks using electric vehicle charging.
It proposes to preventively segment the cyber infrastructure that charging station operators (CSOs) use to communicate with and control their charging stations, thereby limiting the impact of successful cyber-attacks. 
Using real German charging station data and a reconstructed transmission grid model, a threat analysis shows that without segmentation, the successful hack of just two CSOs can overload two transmission grid branches, exceeding the N-1 security margin and necessitating defense measures. 
A novel defense design problem is then formulated that minimizes the number of imposed segmentations while bounding the number of branch overloads under worst-case attacks. 
The resulting IP-MILP bi-level problem can be solved with an exact column and constraint generation algorithm and with heuristics for fast computation on large-scale instances. 
For the near-real-world Germany case, the applicability of the heuristics is demonstrated and validated under relevant load and dispatch scenarios.
It is found that the simple scheme of segmenting CSOs evenly by their installed capacity leads to only 23\% more segments compared to the heuristic  optimization result, 
suggesting potential relevance as a regulatory measure.

\end{abstract}


\thanksto{\noindent This work was sponsored by the German Federal Ministry of Research, Technology and Space in the project CyberStress, funding no. 13N16626.}

\section{Introduction}\label{sec:intro}

The increasing number of electric vehicles (EVs) is accompanied by a rising number of electric vehicle charging stations (EVCSs).
Inevitably, 
the cyber infrastructure that charging station operators (CSOs) use to communicate with and control their EVCSs 
may have various cyber vulnerabilities that can be utilized in a cyber attack to gain control of the EVCSs operation \cite{kaur2025cybersecurity}.
Thereby, an adversary could target the power grid via load-altering attacks (LAA), potentially jeopardizing the safe operation of the power grid \cite{Acharya_Access_2020}.
System operators and regulators, therefore, need to quantify this threat and devise strategies to mitigate the impact of security breaches at CSOs. 

\subsection{Related literature}
LAAs on the power system can be established with high-wattage IoT devices and can target the power grid's safe operation by different means. 
Adversarial load changes and corresponding generation adjustments can lead to line overloads \cite{soltan2018blackiot}, e.g., with EVCSs \cite{khan2019impact}. 
Alternatively, EV-based LAAs may lead to frequency-band violations and generator tripping,
as has been shown for New York City utilizing publicly available data
\cite{Acharya_2020,Acharya_2024}. 
In a different threat instance, the repeated load-switching of hacked EVCSs can excite inter-area oscillation modes in the power grid and thereby lead to subsynchronous stability challenges \cite{wei2023cyber,Abazari_2025}.
While \cite{Acharya_2020,Acharya_2024} demonstrate the threat posed by EV-based LAAs in a realistic setting, an impact analysis on nationwide transmission systems remains open.
Even if the frequency containment reserve (FCR) on a nation-scale  exceeds the manipulable load of a LAA and thereby mitigates frequency stability challenges, overloaded lines and transforms, may still pose a significant problem for such systems.



In response to LAAs, three types of counter measures have been discussed:
Reference \cite{Lakshminarayana_21} identifies critical loads for defense through eigenvalue sensitivity analysis, aiming to mitigate frequency instability induced by dynamic LAAs, but  without prescribing specific counter actions.
Corrective physical measures, such as disconnecting buses with compromised EVCSs or deploying operator-controlled storage, are investigated in \cite{wei2023cyber}.
While physical measures can be effective, their deployment can be more expensive than cyber defense ones, limiting their applicability in reality.
Cyber security counter measures for EVCSs' cyber infrastructure are discussed in \cite{Bhattacharya}.
Segmenting the cyber infrastructure of EVCSs is proposed in \cite{kaur2025cybersecurity} as a preventive defense measure, mitigating the risk of an adversary gaining control over multiple EVCS in case of a successful attack. 
In general, network segmentation is a best practice in the cyber security of various critical infrastructure \cite{enisa_25}. 


While different previous works mention the segmentation of the cyber infrastructure as a counter measure, 
most of them, specifically for EVCSs see \cite{kaur2025cybersecurity}, 
remain confined to the cyber layer,
without quantifying the physical impacts of successful attacks or the defensive measures.
One exception is 
\cite{Arguello_2023}, considering a preventive segmentation of the communication network used by transmission system operators for switching operations.
In this work, the communication network is segmented in such a way that the necessary load shed is minimized in case of a successful hack that leads to an opening of relays in substations. 
The simulations in \cite{Arguello_2023} point to the effectiveness of such network segmentation, as the measure decreases the worst-case attack's effect by up to \SI{50}{\percent} in terms of reducing load shedding.
The effect of a preventive segmentation on mitigating the threat of LAAs has not been investigated so far.




\subsection{Scope \& contributions}

In this work, we show that security breaches of CSOs' cyber infrastructure, used to communicate with and control EVCSs, can harm the power grid under realistic conditions. 
To mitigate this threat, we devise a defense strategy based on segmentation of the CSOs' cyber infrastructure, show how to compute it, and validate its efficacy and efficiency on the near-real-world example.
Specifically, our main contributions are:
\begin{enumerate}
    \item We conduct a threat analysis of EV-based LAAs in a near-real-world case of Germany:
    we use a replication of the German transmission system and real public charging infrastructure data, capturing both the spatial and company-based distribution of EVCS. 
    We focus on overloaded lines and transformers, i.e., branch overloads, as the attack target.
    This is because branch overloads can result from LAAs of smaller size than needed to exceed the FCR and thereby obtain frequency violations. These attacks are also more difficult to analyze compared to exceeding FCR-limits.
    \item We propose a bi-level defense design problem for regulator-enforced, cost-aware segmentation of CSOs' cyber infrastructure, devised to limit the impact of worst-case LAAs to an acceptable number of branch overloads. 
    \item We devise an exact column and constraint generation algorithm to solve the resulting IP-MILP bi-level problem and show its applicability to the IEEE RTS 24-Bus system. In addition, we develop several heuristics to handle large-scale instances
    and benchmark them against the exact approach on the IEEE RTS 24-Bus system.
    \item We demonstrate the effectiveness of the heuristics for defending the near-real-world Germany case.
    We find that a small number of segmentations can ensure that the number of overloaded branches remains within acceptable limits under relevant load and dispatch scenarios.
    \item Our key finding is, that the simple scheme of segmenting CSOs evenly by their installed capacity leads to only \SI{23}{\percent} more segments compared to the heuristic optimization result on the near-real-world case of Germany. This may prove the regulatory relevance of the simple scheme, as it shows a potentially acceptable \textit{price of simplicity}.
\end{enumerate}


\smallskip

Section \ref{sec:problem} introduces the threat and defense setting and proposes a mathematical formulation of the defense design problem. 
Section \ref{sec:sol_meth} describes the exact and heuristic methods to solve the resulting optimization problem. 
Numerical case studies on the IEEE RTS 24-Bus system and the near-real-world case for Germany are provided in section \ref{sec:case}. 
Conclusions are drawn in section \ref{sec:conc}. 
\section{Defense design problem}\label{sec:problem}

\subsection{Problem description}

This work aims to defend the power grid against an adversary intruding the cyber infrastructure of CSOs so as to launch LAAs to overload transmission lines and transformers, i.e., transmission grid branches.
Once the cyber infrastructure is intruded, 
the adversary may terminate the charging of EVs or may send signals to plugged-in EVs to either start charging or even inject power back into the grid. 
As a defense, the regulator may preventively prescribe the segmentation of the cyber infrastructure and the attached EVCSs, to reduce the risk of an adversary's lateral movement following a successful intrusion. This segmentation can mitigate the threat of branch overloads successfully, as shown for an example in Fig. \ref{fig:cso_prob}.


\begin{figure}[t]
  \centering
  \includegraphics[width = 1\linewidth]{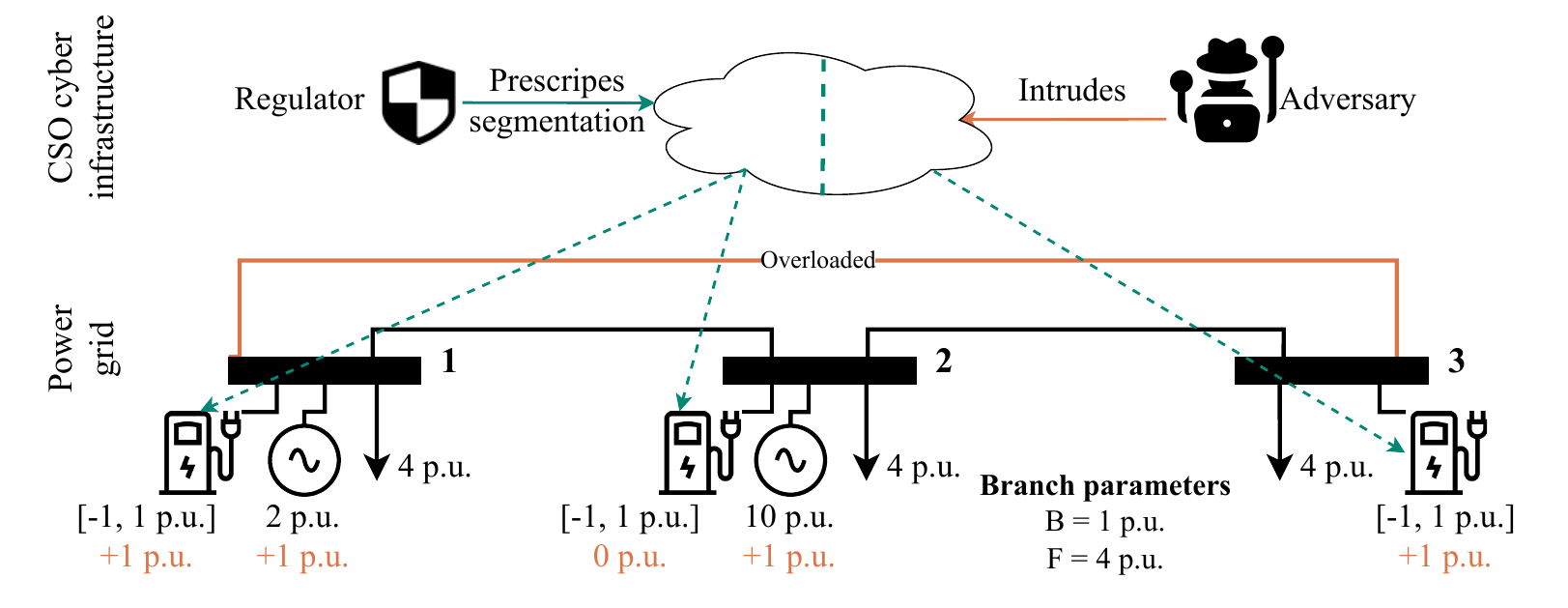}
  \caption{
  By hacking a CSO's cyber infrastructure, an attacker can change default power system operation (black numbers) by a certain level (orange numbers). 
  EVCSs can be manipulated up- and downwards by starting additional charging processes or terminating current ones.
  Attached EVCSs are addressed directly by the adversary, power plants taking part in FCR react indirectly.  
  A branch overload (orange line) follows. 
  A defense against this attack is the segmentation and shown assignment of the CSO's cyber infrastructure (dashed, green lines): in this case, an attack using any one of the segments cannot achieve an overload anymore. The defense is done preventively, i.e., before any attack arises, to prevent a lateral movement of an adversary after an intrusion.
  }
  \label{fig:cso_prob}
\end{figure}




In the following, we formulate a defense design problem to retrieve an effective and efficient segmentation.
The problem minimizes the number of necessary segments while ensuring a predefined acceptable number of overloaded lines and transformers, i.e., branch overloads, in case of a worst-case attack. 
We emphasize that a defense strategy 
applied in practice needs to effectively mitigate the threat in a cost-optimal manner as it increases the system costs. 
In practice, the segmentation of the cyber infrastructure can be implemented by different means, e.g., usage of virtual private networks, splitting of the operator personnel, introducing air gaps in the cyber infrastructure, or requiring independent hard- and software developments. Further, the segmentation should pose no intersection of vulnerabilities across the segments to be effective.
Depending on the required security level, different means may be appropriate, resulting in highly differing costs of additional segments.

\subsection{Mathematical formulation}

We model the defense design problem as a bi-level optimization problem, where 
the upper-level represents the segmentation decisions
and the lower-level the actions of the omniscient adversary aiming at a worst-case LAA given the segmentation.

Let the power grid be described by buses (nodes) $\mathcal{N}$ and branches  $\mathcal{L}$. 
Given the set of operators $\mathcal{O}$, 
$\mathcal{N}_o \subseteq \mathcal{N}$ denotes the set of buses where operator $o \in \mathcal{O}$ controls EVCSs.
Moreover, for operator $o \in \mathcal{O}$, we denote by $\mathcal{S}_o$ the set of possible cyber segments, 
by $b_{o,s} \in \{0,1\}$ the actual use of segment $s \in \mathcal{S}_o$,
and by $a_{o,n,s} \in \{0,1/\text{D},\ldots,1\}$ the fraction of EVCS capacity assigned to segment $s \in \mathcal{S}_o$ at bus $n \in \mathcal{N}_o$.
Note that we use a discretization of the assignment fractions into $\text{D}+1$ values 
to ensure a finite number of overall assignments.
An operator-specific discretization can be modeled analogously. 
Bold variables will denote the set of the corresponding indexed variables throughout.
Further, $u_l^{P,*}\in \{0,1\}$ and $u_l^{N,*}\in \{0,1\}$ indicate the overloading of branch $l \in \mathcal{L}$ by an adversary, whose set of overload maximizing attack decisions, that depend on the  segmentation $\mathbf{a}$, are denoted by  $\mathcal{U}(\mathbf{a})$.

The defense design problem then follows as the bi-level optimization problem
\begin{subequations}\label{eq:seg_problem}
\allowdisplaybreaks
    \begin{align}
    \min_{\mathbf{a},\mathbf{b},\mathbf{u}^{P,*}, \mathbf{u}^{L,*}}  & \; \sum_{o \in \mathcal{O}} \sum_{s \in \mathcal{S}_o} b_{o,s}, \label{eq:seg_prob_obj} \\
    & \text{s.t.} \notag \\
    & a_{o,n,s} \le b_{o,s}, \forall o \in \mathcal{O}, \forall n \in \mathcal{N}_o, \forall s \in \mathcal{S}_o,\label{eq:seg_prob_assign_segment} \\
    & \sum_{s \in \mathcal{S}_o} a_{o,n,s} = 1, \forall o \in \mathcal{O}, \forall n \in \mathcal{N}_o, \label{eq:seg_prob_must_assigned} \\
    & (\mathbf{u}^{P,*}, \mathbf{u}^{L,*}) \in \mathcal{U}(\mathbf{a}),
    \label{eq:seg_prob_sol_adv}   \\
    & \sum_{l \in \mathcal{L}} (u^{P,*}_l + u^{L,*}_l) \le \text{K}. 
    \label{eq:seg_prob_sec_con} 
    \end{align}
\end{subequations}    
The segmentation, potentially done by a regulator, should be cost-efficient, hence, the objective \eqref{eq:seg_prob_obj} minimizes the sum of used segments for all operators. 
Constraint \eqref{eq:seg_prob_assign_segment} ensures that a charging station can only be assigned to used segments, and \eqref{eq:seg_prob_must_assigned} enforces for each operator $o \in \mathcal{O}$ the assignment of the full EVCS capacity at each bus $n \in \mathcal{N}_o$ to the available segments in $\mathcal{S}_o$. 
Constraint \eqref{eq:seg_prob_sol_adv} requires the overload indicators to be optimal solutions of the adversary's lower-level problem parametrized in the assignment $\mathbf{a}$ of EVCSs to segments.
Constraint \eqref{eq:seg_prob_sec_con} expresses that the defense design problem aims to choose segmentation $\mathbf{a}$ such that the worst-case number of overloads is smaller than a given limit $\text{K}$.
If the system is operated in N-1 mode and, hence, one branch can be overloaded anytime without detrimental system effects, then $\text{K}$ would be chosen as one.

The lower-level problem of the adversary is aiming to maximize the number of overloads by hacking certain segments and is defined next.
For operator $o \in \mathcal{O}$, let $h_{o,s} \in \{0,1\}$ indicate the hacking of segment $s \in \mathcal{S}_o$ and $l^{pos}_{o,n} \in \mathbb{R}^{\ge0}$ and $l^{neg}_{o,n}\in \mathbb{R}^{\ge0}$ the change of the charging set-points of the EVCSs at bus $n \in \mathcal{N}_o$.
This changes the power flow  $f_{l} \in \mathbb{R}, \forall l \in \mathcal{L}$, described by the following variables and parameters:
$\theta_n \in \mathbb{R}$ denotes the voltage phase angle in a DC power flow model at bus $n\in\mathcal{N}$, $\mathcal{G}_n$ the set of generators at that bus, and 
$f^D \in \mathbb{R}$ the global frequency deviation. 
Jointly, we denote the set of all lower-level decision variables as $\mathbf{c} := (\mathbf{u}^P, \mathbf{u}^N, \mathbf{h}, \mathbf{l}^{pos}, \mathbf{l}^{neg}, 
\mathbf{f},f^D, \boldsymbol{\theta})$.

The worst-case overloads and the corresponding hacking decisions are then described by the following adversary problem, given the assignment decisions $\mathbf{a}$ of EVCS to segments,
\begin{subequations}\label{eq:adv_problem}
\allowdisplaybreaks
    \begin{align}
    \max_{\mathbf{c}} & \sum_{l \in \mathcal{L}} (u^P_{l}+u^N_{l}), \label{eq:adv_problem_obj} \\
    & \text{s.t.} \notag \\
    & \sum_{o \in \mathcal{O}}\sum_{s \in \mathcal{S}_o} h_{o,s} \le \text{C}^\text{Hack}, \label{eq:adv_prob_budget} \\
    & \text{ \eqref{eq:adv_prob_laa_pos} and \eqref{eq:adv_prob_laa_neg}} \, \forall o \in \mathcal{O}, \forall n \in \mathcal{N}_o: \notag \\
    & l^{pos}_{o,n} \le \text{L}_{o,n}(1-\text{C})\text{C}^\text{ACT} \sum_{s \in \mathcal{S}_o} h_{o,s}\, a_{o,n,s}, \label{eq:adv_prob_laa_pos}\\
    & l^{neg}_{o,n} \le \text{L}_{o,n} \text{C}(1+\text{C}^\text{V2G}) \sum_{s \in \mathcal{S}_o} h_{o,s}\, a_{o,n,s}, \label{eq:adv_prob_laa_neg} \\
    & -\text{LAA}^{\text{max}} \le \sum_{o \in \mathcal{O}}\sum_{n \in \mathcal{N}_o} (l^{pos}_{o,n} - l^{neg}_{o,n}) \le \text{LAA}^{\text{max}}, \label{eq:adv_prob_laa_sum} \\
    & \sum_{g \in \mathcal{G}_n} (\text{P}^{\text{G}}_g - \text{K}^{\text{G}}_g \, f^D
    ) - \sum_{o \in \mathcal{O}} (\text{C}\,\text{L}_{o,n}+l^{pos}_{o,n}-l^{neg}_{o,n})  \notag \\ 
    & \quad -\text{P}^{\text{D}}_n = \sum_{l \in \mathcal{L}} \text{I}_{n,l} f_l, \forall{n} \in \mathcal{N}, \label{eq:adv_prob__balance} \\
    & \sum_{n \in \mathcal{N}} \text{S}_{l,n}\theta_n = f_l, \forall l \in \mathcal{L}, \label{eq:adv_prob_flows} \\
    & f^D = -\frac{\sum_{o \in \mathcal{O}}\sum_{n \in \mathcal{N}_o} (l^{pos}_{o,n}-l^{neg}_{o,n})}{\sum_{n \in \mathcal{N}} \sum_{g \in \mathcal{G}_n} \text{K}^{\text{G}}_g}, \label{eq:adv_prob_f_dev} \\
   & f_{l} - \text{F}_{l} \leq u^P_{l} \text{M}, \forall l \in \mathcal{L},\label{eq:adv_prob_ov_either}\\
   & f_{l} - \text{F}_{l} \geq ( u^P_{l} - 1) \text{M}, \forall l \in \mathcal{L}, \\
   & -f_{l} - \text{F}_{l} \leq  u^N_{l} \text{M}, \forall l \in \mathcal{L},\\
   & -f_{l} - \text{F}_{l} \geq (u^N_{l} - 1) \text{M}, \forall l \in \mathcal{L}. \label{eq:adv_prob_ov_ind} 
    \end{align}
\end{subequations} 

The objective \eqref{eq:adv_problem_obj} of the adversary is to maximize the number of branch overloads.
The adversary's hacking decision is limited by a budget of at most $\text{C}^\text{Hack}$ compromised  segments in \eqref{eq:adv_prob_budget}. 
Constraints \eqref{eq:adv_prob_laa_pos} and \eqref{eq:adv_prob_laa_neg}, respectively, 
limit the positive and negative LAA power to plausible values.
Here, 
$\text{L}_{o,n}$ denotes the installed capacity of EVCSs of operator $o \in \mathcal{O}$ at bus $n \in \mathcal{N}_0$,
$\text{C}$ is the fraction of EVCS capacity that typically charges at any point in time (the charging coincidence),
$\text{C}^\text{ACT} \in [0,1]$ is the fraction of typically non-charging EVCS capacity that could be activated by a LAA,
and $\text{C}^\text{V2G} \in [0,1]$ is the fraction of typically charging capacity that could potentially be ordered to inject power from the vehicle back into the grid.
Constraint \eqref{eq:adv_prob_laa_sum} bounds the aggregated net load change due to the LAA .
Since we focus on branch overloads here, not frequency band violations, the limit $\text{LAA}^{\text{max}}$ would typically be the available FCR power.
The nodal power balance \eqref{eq:adv_prob__balance} for node $n\in\mathcal{N}$ accounts 
for the generation $\text{P}^{\text{G}}_g$ scheduled without an attack and the FCR response 
and the FCR response with gain $K^\text{P}_g$ for each generator $g \in \mathcal{G}_n$, 
the EVCS loads and a base load $\text{P}^{\text{D}}_n$, 
and the grid power injection that is linked to branch power flows via the bus-branch incidence matrix $\text{I}_{n,l}$, $n\in\mathcal{N}$, $l\in\mathcal{L}$.
Constraints \eqref{eq:adv_prob_flows} describes the DC power flow in terms of the branch-bus susceptance matrix $\text{S}_{l,n}$, $l\in\mathcal{L}$, $n\in\mathcal{N}$.
Constraints \eqref{eq:adv_prob_f_dev} models the global frequency deviation depending on the overall power balance.
It is assumed that the power balance was satisfied prior to the LAA.
The branch overload indicators are modeled in constraints \eqref{eq:adv_prob_ov_either}-\eqref{eq:adv_prob_ov_ind}
using a large big-M constant $\text{M}$. 
If a power flow $f_l$ is larger or equal to a predefined threshold $\text{F}_{l}$, $l\in\mathcal{L}$,
the branch is assumed to be overloaded. The threshold may plausibly represent the temporary attainable loading of lines and transformers, necessitating corrective actions within short time, if approached, to maintain safe operation.
As power flows are directed, we model positive and negative overloads separately.



\section{Solution methodology}\label{sec:sol_meth}
The defense design problem is a bi-level optimization problem where all upper level variables are integer, the lower-level variables are both integer and continuous, and all constraints are linear in the variables of the corresponding level.
The IP-MILP bi-level problem can, in principle, be solved by enumerating all possible upper level decisions and solving the corresponding lower level problem.
Yet, more efficient approaches are needed.
To this end, we first devise a column and constraint generating (CCG) algorithm.
Additionally, we develop four segmentation heuristics for large-scale instances. 

Before describing these algorithms, 
we define an additional constraint for problem \eqref{eq:seg_problem} that does not change the value of the optimal solution, but "cuts" a large number of equivalent, permuted assignment solutions:
Assuming that the sets $\mathcal{S}_o$ are ordered and indexed by natural numbers for $o \in \mathcal{O}$, let
\begin{align}
    & b_{o,s} \le b_{o,s-1}, \forall o \in \mathcal{O}, \forall s \in \mathcal{S}_o: s > 1.\label{eq:mp_permut}
\end{align}
A segment can then only be used if its predecessor is already used.
If, e.g., we want to use 2 segments out of 10, the number of equivalent assignments is reduced from $\binom{10}{2}$ to one. 

\subsection{Column and constraint generation}

The CCG algorithm consists of an iteratively growing master-problem $\mathcal{MP}_i$ and a sub-problem $\mathcal{SP}(\mathbf{a_i})$ to determine the optimal attack to consider additionally.

Specifically, let the master problem $\mathcal{MP}_i$ in iteration $i > 1$ with $\eta \in \mathbb{R}$ be given as
\begin{subequations}\label{eq:master_problem}
\allowdisplaybreaks
    \begin{align}
    \min_{\mathbf{a},\mathbf{b},\eta, \mathbf{c}_{\{1, \dots, i\}}}  &\eqref{eq:seg_prob_obj} \label{eq:mp_obj} \\
    \text{s.t. } & \eqref{eq:seg_prob_assign_segment}, \eqref{eq:seg_prob_must_assigned},  \eqref{eq:mp_permut}, \label{eq:mp_seg_cons} \\
        & 0 \le \eta \le \text{K}, \label{eq:mp_ovs} \\
    &\mathcal{C}(\mathbf{a},\mathbf{c}_j,\eta;\boldsymbol{\alpha}_{j-1}),    \forall j \in\{1, \dots, i\} \label{eq:mp_con_set}.
    \end{align}
\end{subequations}
Here, $\mathcal{C}(\mathbf{a},\mathbf{c},\eta;\boldsymbol{\alpha})$ denotes a set of constraints for 
segmentation variables $\mathbf{a}$, threshold variable $\eta$, and lower-level decision variables $\mathbf{c}$
and for fixed attack values 
$\boldsymbol{\alpha} =(\mathbf{h}^f,\mathbf{l}^{pos,f},\mathbf{l}^{neg,f})$ that are determined in previous iterations of the CCG algorithm.
Superscript $f$ denotes fixed variable values throughout. 
The constraints are 
\begin{subequations}\label{eq:ccg_set_i}
\allowdisplaybreaks
    \begin{align}
    \mathcal{C}&(\mathbf{a},\mathbf{c},\eta; \boldsymbol{\alpha}) := \{ \notag\\
    & l^{pos}_{o,n} = \text{l}^{pos,f}_{o,n} \sum_{s \in \mathcal{S}_o} \text{h}^f_{o,s}\, a_{o,n,s}, \, \forall o \in \mathcal{O}, \forall n \in \mathcal{N}_o \label{eq:ccg_laa_pos}\\
    &  l^{neg}_{o,n} =  \text{l}^{neg,f}_{o,n} \sum_{s \in \mathcal{S}_o}\text{h}^f_{o,s}\, a_{o,n,s}, \, \forall o \in \mathcal{O}, \forall n \in \mathcal{N}_o \label{eq:ccg_laa_neg} \\
    & \eqref{eq:adv_prob__balance}-\eqref{eq:adv_prob_ov_ind},\; \label{eq:ccg_cons_borrowed} \\
   & \sum_{l \in \mathcal{L}} (u^P_{l} + u^N_{l}) \le \eta \label{eq:ccg_ovl}\}.
    \end{align}
\end{subequations} 
Problem~\eqref{eq:master_problem} generates new columns (variables) $\mathbf{c}_j$ in each iteration that describe the power flow and the overloads,  corresponding the currently, considered segmentation $\mathbf{a}$ and the fixed attack vector $\boldsymbol{\alpha}_{j}$.
The $\mathbf{c}_j$ values are determined by constraint set \eqref{eq:ccg_set_i} instead of the lower level optimization problem \eqref{eq:adv_problem}, as implied by \eqref{eq:seg_prob_sol_adv}.
This renders \eqref{eq:master_problem} not a bi-level but a single-level MILP problem that can be solved with standard solvers.

To understand the rationale behind conditions \eqref{eq:ccg_set_i} suppose that $\boldsymbol{\alpha}$ is determined as a response for assignment $\mathbf{a}^\dagger$, as described below.
If then the argument $\mathbf{a}$ for the conditions \eqref{eq:ccg_set_i} is chosen as $\mathbf{a}^\dagger$, then the number of overloads in any valid vector $\mathbf{c}$ is the same as the one retrieved by solving the lower-level adversary problem \eqref{eq:adv_problem} with the segmentation $\mathbf{a}^\dagger$,
since in this case $l^{pos}_{o,n} = \text{l}^{pos,f}_{o,n}$ and $l^{neg}_{o,n} = \text{l}^{neg,f}_{o,n}$ and the branch overloads are uniquely determined by this.
If $\mathbf{a}$ is not chosen as equal to $\mathbf{a}^\dagger$, then constraints \eqref{eq:ccg_laa_pos} and \eqref{eq:ccg_laa_neg} 
ensure that the attack vector, that is part of any $\mathbf{c}$, is consistent with the considered assignment $\mathbf{a}$.
Moreover, the number of branch overloads, computed by the components of any such $\mathbf{c}$, is always a lower bound for the number of branch overloads retrieved by solving the lower-level adversary problem \eqref{eq:adv_problem} with the same segmentation $\mathbf{a}$.
This is because the lower-level problem \eqref{eq:adv_problem} can deliberately choose the optimal attack, while in conditions \eqref{eq:ccg_set_i} the attack is fixed via the argument $\boldsymbol{\alpha}$. The other conditions are identical and uniquely determine the overloads.

In each iteration $i$ the master problem grows by additionally considering the worst-case attack $\boldsymbol{\alpha}_i$ computed for the previous optimal assignment vector $\mathbf{a}_{i-1}$.
The worst-case attack is determined through the sub-problem $\mathcal{SP}(\mathbf{a})$ given as 
\begin{subequations}\label{eq:sub_problem}
\allowdisplaybreaks
    \begin{align}
\max_{\mathbf{c}} \; &\eqref{eq:adv_problem_obj} \label{eq:sp_obj} \\
    \text{s.t. } &
    \eqref{eq:adv_prob_budget}-\eqref{eq:adv_prob_f_dev}, \label{eq:sp_adv_cons} \\
   & f_{l} - \text{F}_{l}(1+\epsilon) \leq u^P_{l} \text{M}, \forall l \in \mathcal{L}, \label{eq:sp_ov_ind1} \\
   & f_{l} - \text{F}_{l}(1+\epsilon) \geq ( u^P_{l} - 1) \text{M}, \forall l \in \mathcal{L}, \\
   & -f_{l} - \text{F}_{l}(1+\epsilon) \leq  u^N_{l} \text{M}, \forall l \in \mathcal{L},\\
   & -f_{l} - \text{F}_{l}(1+\epsilon) \geq (u^N_{l} - 1) \text{M}, \forall l \in \mathcal{L} .\label{eq:sp_ov_ind4} 
    \end{align}
\end{subequations} 
Sub-problem \eqref{eq:sub_problem} is identical to the lower level \eqref{eq:adv_problem} except for the conditions \eqref{eq:sp_ov_ind1}-\eqref{eq:sp_ov_ind4} that replace conditions \eqref{eq:adv_prob_ov_either}-\eqref{eq:adv_prob_ov_ind} and increase the overload-branch-ratings $\text{F}_{l}$ by a small constant $\epsilon > 0$.
This is done because conditions  \eqref{eq:adv_prob_ov_either}-\eqref{eq:adv_prob_ov_ind}
do not uniquely determine the overload status of branch $l \in \mathcal{L}$ for $f_l = \text{F}_{l}$.
The sub-problem aiming at maximization of overloads would count the branch as overloaded whereas the master problem, containing \eqref{eq:adv_prob_ov_either}-\eqref{eq:adv_prob_ov_ind} via \eqref{eq:ccg_cons_borrowed} and aiming to minimize the overloads, would not.
The slight increase of the overload-branch-ratings in the subproblem solves this inconsistency.

The full CCG algorithm including initialization and convergence checks is given in Algorithm~\ref{alg:ccg}.
Let $\mathbf{a}^{max}$ denote the maximal segmentation of the EVCSs, i.e., every available segment of an operator is used and accommodates only the smallest possible amount of EVCS capacity for each bus.
The number of worst-case branch overloads with this maximal segmentation is a lower bound $\mu^{lo}$ on the number of overloads for any other segmentation with fewer and, thus, larger segments.
If the number exceeds the allowed limit $\text{K}$, no acceptable defense exists and the algorithm terminates. 
A natural upper bound $\mu^{lo}$ on the number of branch overloads is the total number of branches.
While the upper bounds is larger than the lower bound, we then alternate between computing a new worst-case attack for the current segmentation and updating the segmentation as to defend against the newly discovered attack as well as all previous ones.
As initial segmentation $\mathbf{a}_0$ we use
the minimal segmentation where all EVCSs of an operator are assigned to a single segment, i.e., no preventive segmentation is enforced. 
For this minimal segmentation, subproblem $\mathcal{SP}$ returns the largest amount of overloaded branches possible. 

\begin{algorithm}\label{alg:ccg}
\caption{Column and constraint generation}
Initialize iteration counter $i \gets 0$\;
Solve $\mathcal{SP}(\mathbf{a}^{max})$ and initialize lower bound $\mu^{lo} \gets \sum_{l\in \mathcal{L}} u^P_{l} + u^N_{l}$\;
\If{$\mu^{lo} > \text{K}$} {
\Return \textit{Problem is infeasible.}}
Initialize upper bound $\mu^{up} = |\mathcal{L}|$\;
\While{$\mu^{up} > \mu^{lo}$}{
    $i \gets i + 1$ \;
    Solve $\mathcal{SP}(\mathbf{a}_{i-1})$ to 
    obtain attack $\boldsymbol{\alpha}_i$
    and update upper bound $\mu^{up} \gets \sum_{l\in \mathcal{L}} u^P_{l} + u^N_{l}$\;
    Solve growing $\mathcal{MP}_i$ to obtain segmentation $\mathbf{a}_i$ and update lower bound 
    $\mu^{lo} \gets \eta$\;
}
\Return final segmentation $\mathbf{a}_i$
\end{algorithm}

\begin{proposition}\label{prop:conv_ccg}
The proposed CCG algorithm terminates in a finite number of iterations with the exact optimal solution of the original defense design problem \eqref{eq:seg_problem}.
\end{proposition}
\begin{proof}  
The number of iterations is finite since the number of possible segmentations is finite and the master problem cannot produce the same segmentation twice without terminating the algorithm;
if it would, the worst-case attack generated by $\mathcal{SP}(\mathbf{a}_{i-1})$ would be identical to an attack already considered in the master problem 
and then $\mu^{lo}$ and $\mu^{up}$ would be identical in the next iteration, terminating the algorithm.
Given the termination of the CCG algorithm with a repeated segmentation vector $\mathbf{a}_{i}$, the number of used segments is minimized by the master-problem and the subproblem guarantees that no other attack would lead to more branch overloads.
\end{proof}

In practice, the number of possible segmentations and corresponding worst-case attacks may be large which means that the CCG algorithm may take long to terminate for large-scale power grids. 
While the CCG algorithm offers at least valid upper and lower bounds on the possible number of branch overloads at any time such that an informed decision about early termination can be done,
we nevertheless devise several faster heuristics in the following.

\subsection{Heuristic approaches}




We propose four heuristic schemes: uniform thresholding, balanced clustering, and two iterative informed schemes.


\subsubsection{Uniform thresholding} 
In this scheme, the number of used segments per operator is determined as the total installed EVCS capacity divided by a predefined maximally allowed capacity per segment.
At each bus, the EVCSs of an operator are uniformly distributed to all available segments.
For the maximally allowed capacity \texttt{CS} per segment, we denote this heuristic as \texttt{uni\_thres\_CS}.


\subsubsection{Balanced Clustering} 
In this scheme, the number of segments per operator is given as a constant \texttt{KS}.
At each bus, all EVCSs of an operator are assigned to exactly one of the available segments.
This is done such that the buses corresponding to one segment are close together, while the total EVCS capacity per segment is also balanced.
This challenge is approached as a capacitated clustering problem that is formulated and solved as an optimization problem for each operator $o \in \mathcal{O}$.
For the binary assignments $a_{o,n,s}$, $s \in \mathcal{S}_o$, the inequality of the segments' sizes $\kappa \in \mathbb{R}$ and a fixed penalty factor $\lambda \in \mathbb{R}^{\ge 0}$,
we solve

\begin{subequations}\label{eq:cap_clus_problem}
\allowdisplaybreaks
    \begin{align}
    \min_{\mathbf{a},\kappa}  & \sum_{s \in \mathcal{S}_o} \sum_{n,k \in \mathcal{N}_o} a_{o,n,s}a_{o,k,s} \text{d}_{n,k} + \lambda \kappa\label{eq:cap_clus_obj} \\
    & \text{s.t.} \sum_{s \in \mathcal{S}_o} a_{o,n,s} = 1,\forall n \in \mathcal{N}_o  \label{eq:cap_assign}\\
    & \sum_{n \in \mathcal{N}_o}  \text{L}_{o,n} a_{o,n,s} - \frac{1}{\texttt{KS}} \sum_{n \in \mathcal{N}_o} \text{L}_{o,n} \le \kappa, \forall s \in \mathcal{S}_o \label{eq:cap_knap}.
    \end{align}
\end{subequations}  
As distance $\text{d}_{n,k}$ between buses $n$ and $k$ we use the electrical distance defined in \cite{Cotilla_2013}. 
The bi-linearity in \eqref{eq:cap_clus_obj} can be linearized exactly for the binary variables  yielding an integer linear problem. 
We denote this heuristic for \texttt{KS} segments  as \texttt{clus\_seg\_KS}.

\subsubsection{Iterative informed}

The iterative schemes start with the minimal segmentation $\mathbf{a}_0$ as the current segmentation $\mathbf{a}$.
Then the sub-problem $\mathcal{SP}(\mathbf{a})$ is evaluated to obtain the maximal number of branch overloads and the worst-case attack for the current segmentation $\mathbf{a}$. 
If the obtained number of overloaded branches is below $\text{K}$, the algorithm terminates.
Otherwise, the hacked segments are split into two segments and the algorithms iterates.
The splitting of the hacked segments can be done either uniformly per bus or via balanced clustering, yielding two heuristics. If the splitting is done uniformly per bus with \texttt{S} splits, the heuristic abbreviates to \texttt{itin\_thres\_S}, and to \texttt{itin\_clus\_KS}, if the balanced clustering approach with \texttt{KS} segments is used. 

\section{Numerical Case Studies}\label{sec:case}

We conduct simulations on the IEEE RTS 24-Bus system and a near-real-world case of the German transmission system. 
For both grids, we set the number $\text{K}$ of allowed overloads in the worst-case to one, representing an operative N-1 security margin. 

We implemented the case studies in \texttt{Julia} using \texttt{JuMP} \cite{jump} and \texttt{PowerModels} \cite{powermodels}. As solver, we use \texttt{Gurobi} 12.0.0 \cite{gurobi} with default settings.  The simulations are run on a laptop with an Intel\textsuperscript{\textregistered} Core\textsuperscript{TM} i7-1260P CPU and \SI{32}{\giga\byte} of RAM. 


\subsection{IEEE RTS 24-Bus system}\label{sec:24_bus}
We first benchmark the heuristics against the exact CCG approach on the IEEE RTS 24-Bus system.\footnote{The data of the IEEE RTS 24-Bus system can be found at  \url{https://github.com/MATPOWER/matpower/blob/master/data/case24_ieee_rts.m}.} 
To this end, we set the test case under stressed operational conditions by reducing the branches' transmission capacity to \SI{65}{\percent}, as in \cite{kuroptev2025stress} and others. Further, we deploy 15 equally large charging stations on the grid that are assigned equally to five CSOs. The discretization parameter $\text{D}=2$ is chosen. The charging load totals to \SI{10}{\percent} of the grid's overall load, and the charging simultaneity is set to \SI{20}{\percent} (C = 0.2), which is accounted for in the generator dispatch that is computed by a DC-OPF economic dispatch. 
The assignment of charging stations to buses is shown in Fig. \ref{fig:ieee_24_bus}, and further parameters are given in Table \ref{tab:simu_params_24bus}.

\begin{table}[tbp]
  \centering
  \caption{Simulation parameters for the IEEE RTS 24-Bus system case.}
  \label{tab:simu_params_24bus}
  \tiny
  \begin{tabular}{@{} llllllll @{}}
    \toprule
    \textbf{$\text{C}^\text{Hack}$} & 
    \textbf{$\text{K}^\text{G}_g$} & 
    \textbf{$\text{C}^\text{V2G}$} & 
    \textbf{$\text{C}^\text{ACT}$} & 
    \textbf{$\text{LAA}^\text{max}$} & 
    \textbf{M} & 
    \boldmath$\epsilon$ \\
    \midrule
    2 & 
    relative to dispatch & 
    0 & 
    1 & 
    0 & 
    100 & 
    $10^{-3}$ \\
    \bottomrule
  \end{tabular}
\end{table}

With a hacking budget $\text{C}^\text{Hack}$ of two and no preventive segmentation, i.e., using only the minimal segmentation $\mathbf{a}_0$,
the worst-case LAA leads to the overloading of two branches, necessitating a defense.

\begin{figure}[t]
  \centering
  \includegraphics[width = 0.8\linewidth]{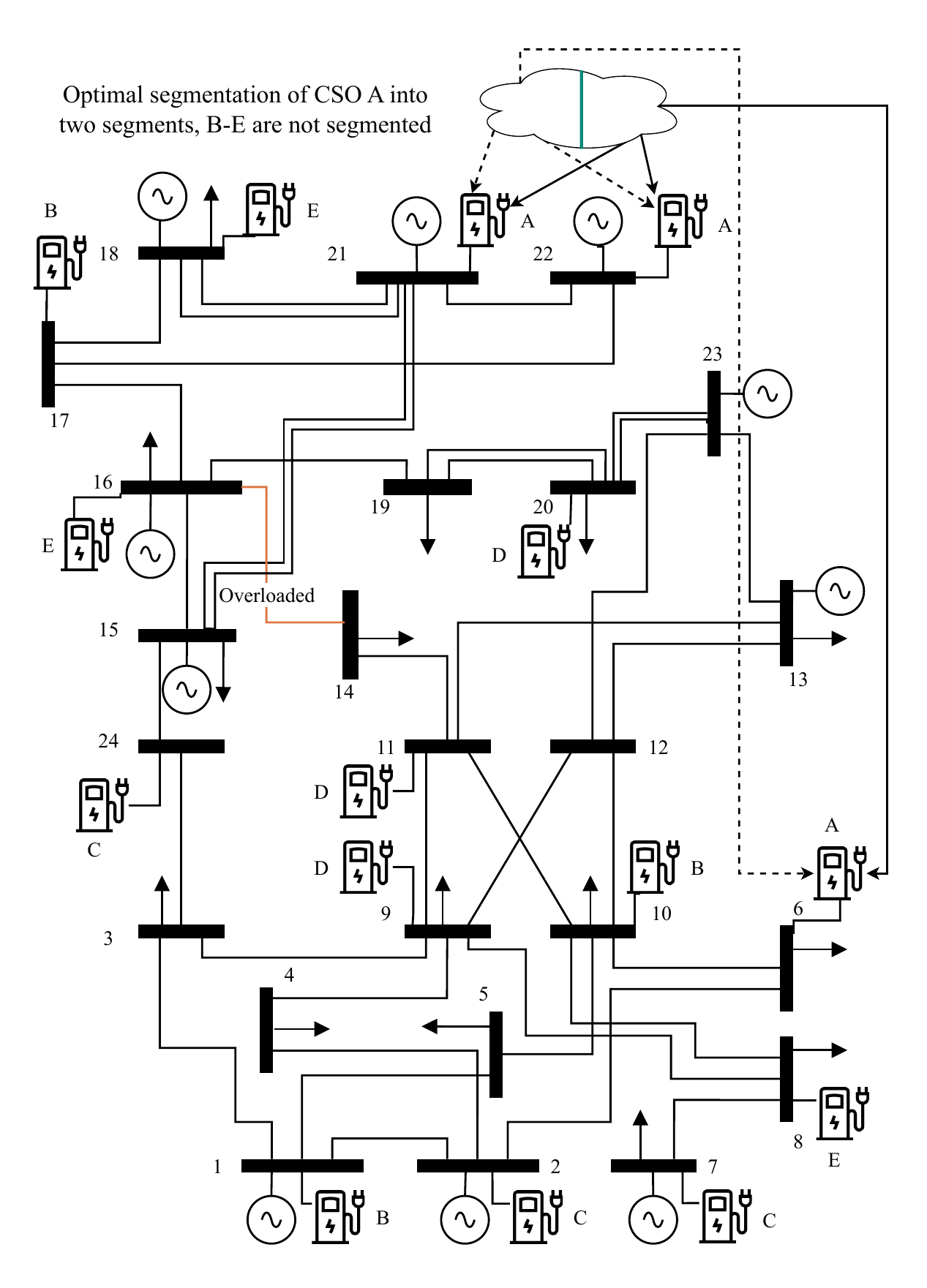}
  \caption{The optimal preventive segmentation of the IEEE RTS 24-Bus system with five CSOs (black letters next to EVCSs)
  splits only the cyber infrastructure of CSO A into two segments, with a uniform distribution at each bus.
  This results in at most one overloaded branch (marked orange) in case of a worst-case LAA attack.
  } 
  \label{fig:ieee_24_bus}
\end{figure}

In Fig. \ref{fig:ieee_24_bus} the optimal segmentation computed with the CCG algorithm is shown. 
It uses 6 segments, i.e., the cyber infrastructure of only one CSO is split.
The observed splitting is identical to a uniform splitting per bus.

Changing the discretization parameter to one resulted in seven required segments, and changing it to three resulted in 6 segments. 
This implies that $\text{D}=2$ is an appropriate discretization choice here and a more detailed discretization does not improve the results.


All heuristics yield a larger number of segments than the exact CCG approach to ensure at most 1 overloaded branch.
Specifically, we obtain 7 segments for  \texttt{itin\_thres\_2} and \texttt{itin\_clus\_2}.
This means, the iterative informed segmentation schemes are \SI{16.7}{\percent} suboptimal, considering the examined setting. The schemes \texttt{uni\_thresh\_0.285}, where a segments capacity is set to 0.285 p.u., i.e., \SI{50}{\percent} of an operator's installed capacity, and \texttt{clus\_seg\_2} require both 10 segments. This renders the simple schemes less efficient than the iterative informed schemes.


\subsection{Near-real-world case: setup}

We built the near-real-world case of Germany using the power grid model ``SciGrid'' from \cite{matke2017structure}
 and the official German public charging infrastructure data from \cite{bnetza_emobilitaet_2025}.\footnote{The ``SciGrid'' data set is retrieved from \texttt{PyPSA} \url{https://doi.org/10.5281/zenodo.14824654}.}

The power grid data from \cite{matke2017structure} contains a reconstruction of the transmissions system topology (\SI{220}{\kilo\volt} and above), branch parameters, loads, and available generators including cost parameters. 
Further, the substation (bus) data include the geographic coordinates and the locally available voltage levels. In total, the ``SciGrid'' consists of 585 buses and 948 branches.
The official German public charging infrastructure data from \cite{bnetza_emobilitaet_2025} contains all public EVCSs with their installed nominal, charging capacities, their geo-coordinates, and the responsible CSOs. 
EVCSs were assigned to the nearest grid bus that has a connection to the distribution grid, indicated by a locally available \SI{110}{\kilo\volt} voltage level.
The resulting grid and the nodal EVCS capacities are shown in Fig. \ref{fig:scigrid_charging_stations}.\footnote{The near-real-world case data sets used in this paper are published at \url{https://github.com/EINS-TUDa/CSO_Segmentation}.}



\begin{figure}[htbp]
  \centering
  \includegraphics[width = 0.8\linewidth]{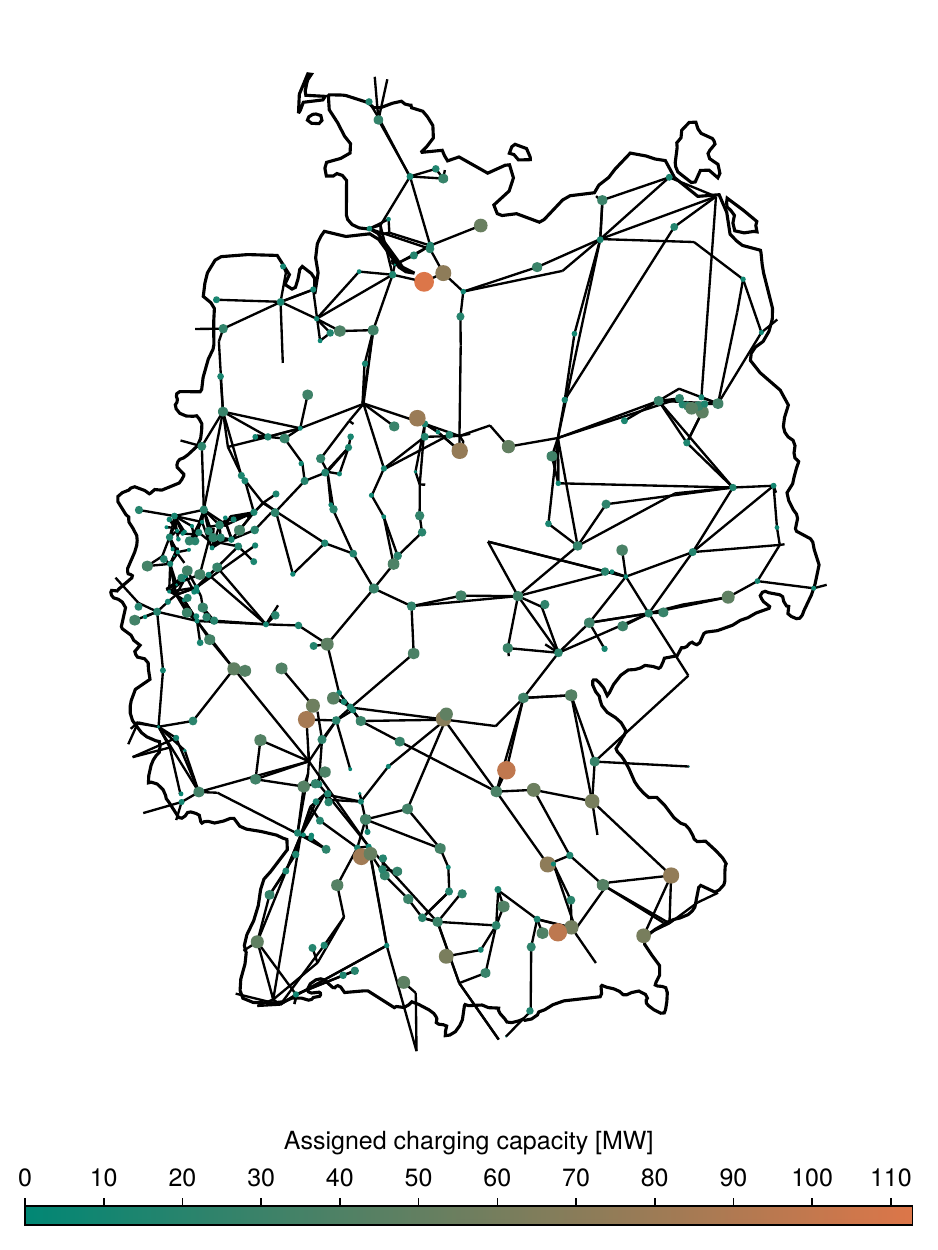}
  \caption{Reconstructed transmission grid of Germany from \cite{matke2017structure} joined with EVCS data of \cite{bnetza_emobilitaet_2025}. 
  The assigned charging capacity to a bus is coded by size and color of the node.} 
  \label{fig:scigrid_charging_stations}
\end{figure}

In our simulations, we analyze four scenarios representing varying load and renewable energy availability conditions throughout the year.  
For the load we define three values, a low load of \SI{42}{\giga\watt} as the base, a medium load as 1.5 times, and the high load as 1.8 times the base load. 
As renewable energy availability, we use \SI{5}{\percent} and \SI{82}{\percent} as the lowest and highest capacity factor for wind and \SI{0}{\percent} and \SI{68}{\percent} for PV.\footnote{The renewable power availabilities are adapted from the peak injections for Germany in 2014 from \url{ https://www.energy-charts.info/charts/power/chart.htm?c=DE&source=public&interval=year&year=2014}.} 
The capacity factor are assumed geographically uniform throughout Germany.
We combine this values into four scenarios: a medium load-high renewable (MLHR), high load-low renewable (HLLR), low load-no PV (LLNP), and low load-low wind  (LLLW). 
The scenarios' resulting load and renewable availability values are given in Table \ref{tab:ta_scenarios}.

The default setpoints $\text{P}^\text{G}$ of the generators are computed by a DC-OPF economic dispatch where the branch limits are set to the permanent attainable loading (PATL). 
For simplicity reasons, we assume that every dispatched generator is contributing to the FCR proportionally to its scheduled value.
As temporary attainable loadings are missing in the data set, we add \SI{5}{\percent} to the PATL values to obtain the values $\text{F}_l$, $l \in \mathcal{L}$, that the adversary must overcome to achieve a overload.

Concerning the modeling of the CSOs, we observe that the size of the CSOs, in terms of installed EVCS capacity, is very unevenly distributed.
The 20 largest of the 11201 CSOs account for \SI{58}{\percent} of the installed EVCS capacity of in total \SI{6.6}{\giga\watt} across Germany.
We therefore explicitly model only the 20 largest CSOs, while summarizing the rest into one non-hackable CSO. 


\begin{table}[tbp]
  \centering
  \caption{Scenarios with their load and renewable power availability and the branch overloads resulting from Worst-Case LAAs.}
  \label{tab:ta_scenarios}
  \begin{tabular}{lcccccc}
    \toprule
     & \textbf{Load} & \textbf{Wind} & \textbf{PV} & \multicolumn{2}{c}{\textbf{Overloaded branches}} \\
    \cmidrule(lr){5-6}
    \textbf{Scenario}  & [GW] & [GW] & [GW] & $\text{C}^\text{Hack}=2$ & $\text{C}^\text{Hack}=10$ \\
    \midrule
    \textbf{MLHR}    & 62 & 33 & 25   & 4 & 10 \\
    \textbf{HLLR}         & 75 &  2 &  0   & 4 & 12 \\
    \textbf{LLNP}           & 42 & 33 &  0   &  4 &  8 \\
    \textbf{LLLW}    & 42 &  2 & 25 & 2 &  6 \\
    \bottomrule
  \end{tabular}
\end{table}

Further parameters for our simulations for the near-real-world case are shown in Table \ref{tab:simu_params}.
The adversary's hacking budget ($\text{C}^\text{Hack}$) of ten may appear high when considering only 20 modeled CSOs. However, the main components of the cyber infrastructure of all CSOs may originate from only a few vendors.
This increases the risk of a security breach affecting multiple CSOs simultaneously. 
The charging coincidence ($\text{C}$), the vehicle to grid availability ($\text{C}^\text{V2G}$), and the charging activation ($\text{C}^\text{ACT}$) factor
are chosen relatively high.
This represents a high utilization of the EVCS-infrastructure and a modern EV-fleet, which can soon be expected due to an ongoing rise of the EVs share.  
The  positive and negative FCR limits are chosen as $\pm \SI{600}{\mega\watt}$, similar to the provisioned FCR for the German load frequency containment block, thereby limiting the magnitude of a considered LAA ($\text{LAA}^\text{max}$).\footnote{The FCR-demand of the German load frequency containment block is \SI{553}{\mega\watt} and published at \url{https://www.entsoe.eu/network_codes/eb/fcr}.}

\begin{table}[hbp]
  \centering
  \caption{Simulation parameters for the near-real-world case.}
  \label{tab:simu_params}
  \tiny
  \begin{tabular}{@{} llllllllll @{}}
    \toprule
    \textbf{$\text{C}^\text{Hack}$} & 
    \textbf{$\text{K}^\text{G}_g$} & 
    \textbf{$\text{F}_l$} & 
    \textbf{C} & 
    \textbf{$\text{C}^\text{V2G}$} & 
    \textbf{$\text{C}^\text{ACT}$} & 
    \textbf{$\text{LAA}^\text{max}$} & 
    \textbf{M} \\
    \midrule
    10 & 
    relative to dispatch & 
    1.05 $\text{PATL}_l$ & 
    0.7 & 
    1 & 
    1 & 
    \SI{600}{\mega\watt} & 
    100 \\
    \bottomrule
  \end{tabular}
\end{table}


\subsection{Near-real-world  case: threat analysis}

We conduct the threat analysis on the near-real-world case for Germany by solving the adversary problem \eqref{eq:adv_problem} without any preventive segmentation of the CSOs' cyber infrastructure, i.e., using the minimal segmentation $\mathbf{a}_0$.

In Table \ref{tab:ta_scenarios} the number of branch overloads for the worst-case LAA, depending on the hacking budget $\text{C}^\text{Hack}$ and the scenario, is shown. 
Several observations can be drawn: 
first, in every scenario and for both hacking budgets the number of overloaded branches exceeds the N-1 security margin, necessitating defensive measures. 
Second, scenario HLLR has more worst-case overloads than the other scenarios.
This may be because in the HLLR scenario the loading of the transmission branches is especially high, reducing the relative attack effort to introduce overloads.
Third, the LLNP scenario has more worst-case overloads than the LLLW scenario.
This shows that not only the load, but also the type of generation dispatch has an impact on the branch overloads.
In the LLNP scenario generation is mostly wind power in Northern Germany.
This is known to frequently necessitate redispatch measures due to a weak north-south grid connection in Germany,
hence, enabling an adversary to achieve more overloads in this scenario as compared to the high PV-availability scenario LLLW. 
In sum, we identify the HLLR scenario as representing the most challenging grid situation. We thus use it in the following to derive the optimal defensive segmentation.



\begin{figure}[ht]
  \centering
  \includegraphics[width = 0.8\linewidth]{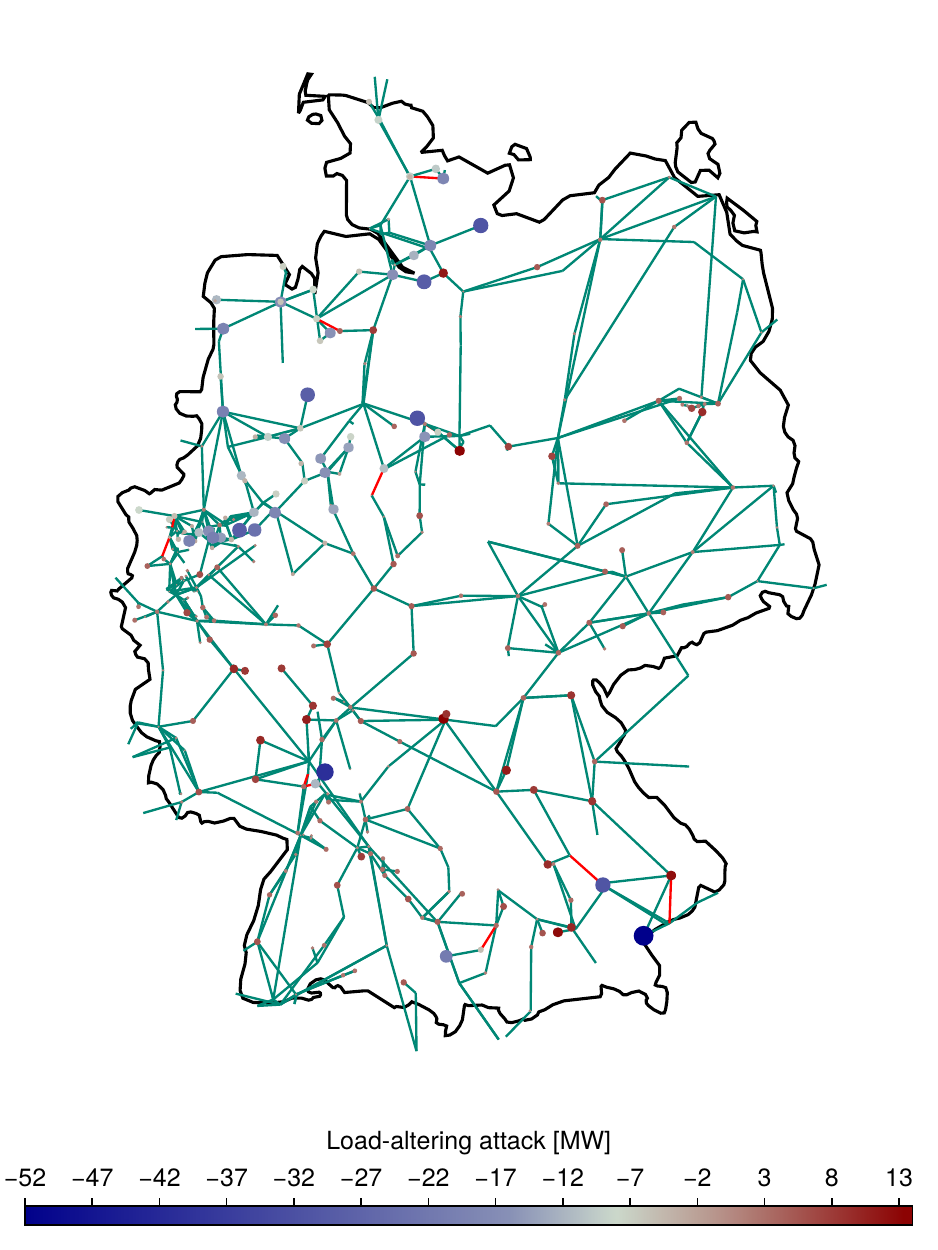}
  \caption{LAA decision and overloaded branches for the HLLR scenario and hacking budget $\text{C}^\text{Hack} = 10$. The absolute load-altering at a bus is coded by the size of the node. 
  Parallel branches are not shown individually; if one is overloaded, all are marked red.} 
  \label{fig:ta_grid}
\end{figure}

Before turning to defense design, we examine the HLLR scenario in more detail in Fig. \ref{fig:ta_grid} where the LAA decision and the overloaded branches are displayed. 
The overloaded branches are located in the highly loaded metropolitan areas in the South-East (Munich area) and the West (Rhine-Ruhr and Rhine-Main area). 
In these areas a load redistribution decision is taken by the adversary, where the load at a few buses is strongly reduced, while it is increased at others, leading to local power flows overloading the branches. 



\subsection{Near-real-world case: defense optimization}

We applied the exact CCG algorithm for defense optimization to all four scenarios, with  varying hacking budgets.
The algorithm did not terminate for this large-scale grid within 100 iterations, with only one exception. 
The LLLW scenario with settings $ \text{C}^\text{Hack}=2 $, $\text{D}=2$, and $\epsilon=10^{-4}$ did terminate after 2 iterations in \SI{294}{\sec}.
The algorithm then finds an exact optimal segmentation with 21 segments, ensuring at most one overloaded branch. 
For comparison, the iterative informed heuristic \texttt{itin\_thres\_2} and \texttt{itin\_clus\_2} need one segment more to ensure at most one overload.

\begin{figure}[ht]
  \centering
 
  \includegraphics[width = 1\linewidth]{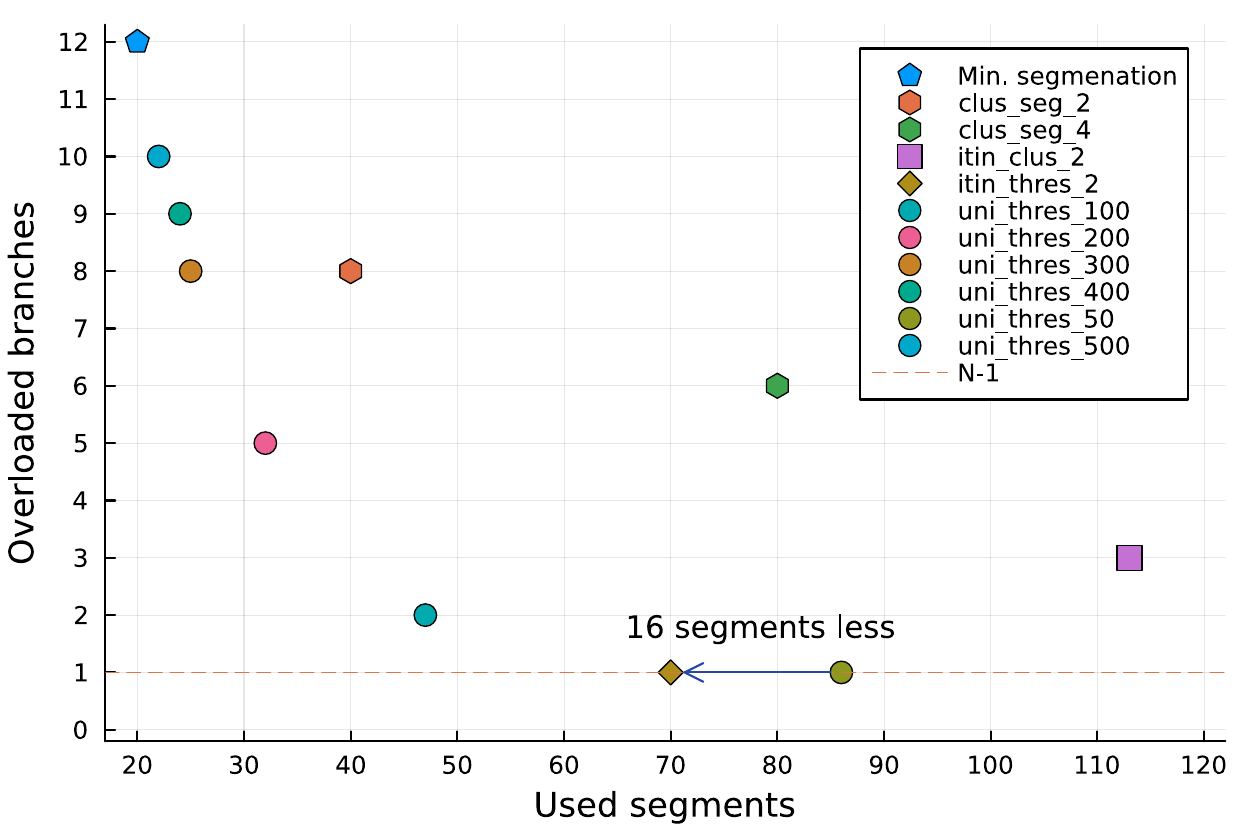}
  \caption{Number of overloaded branches versus the number of used segments for different heuristic segmentation schemes in the scenario HLLR with hacking budget $\text{C}^\text{Hack}=10$. 
  The iterative informed heuristic segmentation \texttt{itin\_thres\_2}, where every hacked segment is uniformly split, ensures at most one overloaded branch in the worst-case, while requiring 16 segments less than the uniform thresholding scheme  \texttt{uni\_thres\_50}. 
  Thresholds \texttt{CS} for the \texttt{uni\_thres\_CS} scheme are given in MW. 
  Note that 20 segments are minimally needed, one for each of the 20 modeled CSOs.} 
  \label{fig:seg_heur}
\end{figure}

Fig. \ref{fig:seg_heur} shows the number of overloaded branches versus the number of used segments for different heuristic segmentation schemes in the scenario HLLR with hacking budget $\text{C}^\text{Hack}=10$.
Three observations are made.
First, the simple uniform thresholding scheme \texttt{uni\_thres\_CS} monotonically reduces the number of overloads and increases the number of used segments with decreasing thresholds.
For $\texttt{CS}=\SI{50}{\mega\watt}$, at most one overloaded branch is ensured and 86 segments are used.
This could inform regulatory measures by indicating that \SI{50}{\mega\watt} is a suitable threshold level to not exceed N-1 margins even in case of LAAs.
Second, we observe that iterative informed scheme \texttt{itin\_thres\_2} can achieve the same security level with 16 segments less than the \texttt{uni\_thres\_50} scheme.
The latter is only \SI{23}{\percent} less efficient in terms of segmentation effort.
Since \texttt{uni\_thres\_50} might be easier to regulate and implement than \texttt{itin\_thres\_2},
this quantifies the \textit{price of simplicity}.
Third, the balanced clustering using schemes \texttt{clus\_seg\_2}, \texttt{clus\_seg\_4}, and \texttt{itin\_clus\_2} perform less effectively and efficiently. For these the penalty parameter of problem \eqref{eq:cap_clus_problem} was set to $\lambda = 10^5$. In particular, \texttt{clus\_seg\_2} needs 15 segments more than \texttt{uni\_thres\_300} to achieve at most 8 overloaded branches, and the iterative informed scheme \texttt{itin\_clus\_2} requires 113 segmentations without ensuring at most one overloaded branch in the worst-case after ten iterations of the algorithm. 
As shown in Fig. \ref{fig:ta_grid}, LAAs on the Germany grid favor local redistribution, which is not hindered by cluster-based segmentation scheme, as nearby located EVCS are grouped into the same segments.

Table \ref{tab:def_scenarios} shows how the segmentation 
that is derived with the \texttt{itin\_thres\_2} scheme for the HLLR scenario with $\text{C}^\text{Hack}=10$, see Fig. \ref{fig:seg_heur},
performs under worst-case attacks on the other scenarios.
Additionally, we evaluate the effectiveness of \texttt{uni\_thres\_100} since it leads to only two overloads for high-stress HLLR scenario and may perform better in less stressed grid situations.
The iterative informed scheme \texttt{itin\_thres\_2} ensures zero or one overloaded branch in the worst-case for the remaining scenarios.
This indicates the robustness of the resulting segmentation to changing grid conditions. 
The uniform thresholding scheme \texttt{uni\_thres\_100} also ensures at-most one overloaded brach for the low load scenarios, but fails to achieve this for higher loads.

\begin{table}[htbp]
  \centering
  \caption{Sensitivity analysis of the preventive segmentations devised on the HLLR grid-scenario.}
  \label{tab:def_scenarios}
  \begin{tabular}{lccc}
    \toprule
     & \multicolumn{2}{c}{\textbf{Overloaded branches using segmentation of}} \\
    \cmidrule(lr){2-3}
    \textbf{Scenario} &\texttt{itin\_thres\_2} & \texttt{uni\_thres\_100} \\
    \midrule
    \textbf{MLHR}    & 1 & 2 \\
    \textbf{LLLP}       &  0 &  1 \\
    \textbf{LLLW}    & 0 &  1 \\
    \bottomrule
  \end{tabular}
\end{table}




\section{Conclusion}\label{sec:conc}

In this paper we have proposed a preventive defensive segmentation of the CSOs' cyber  infrastructure to mitigate the threat of overloaded branches, in particular lines and transformers, posed by LAAs utilizing EVCSs. 
We have modeled the defense design problem as a bi-level problem, minimizing the number of necessary segments, while limiting the maximum number of overloaded branches for the worst-case LAA. 
For this bi-level optimization problem, we derived an exact CCG algorithm and heuristic solution approaches, 
which we assessed on the IEEE RTS 24-Bus system and a near-real-world case of Germany. 
The threat analysis on the near-real-world case of Germany indicates a violation of the N-1 security margin, if at least two operators are compromised. 
The evaluation of the designed defensive segmentations for this case indicates that the iterative informed heuristic using uniform splitting scheme is the most efficient heuristic, while a simpler scheme limiting the charging capacity of a segment to \SI{50}{\mega\watt} also ensures one overloaded branch, while needing 16 segments more - the \textit{price of simplicity}.
As the threat analysis has shown the necessity to design defense measures, we hope that our findings on practical, heuristic segmentation schemes inform regulatory measures prescribing a segmentation of the CSOs' cyber infrastructure to mitigate the imposed threat.


Future work can be devoted to mitigate the potential threat imposed to the power grid's voltage stability, when EVs are used in a LAA. As an adversary could manipulate the power factor of the charging process by hacking the EVs' on-board-chargers and therefore perturb the reactive power balance of the power system potentially leading to a voltage collapse. Furthermore, future work could transfer the preventive segmentation approach to mitigate threats arising from LAAs that utilize other devices than EVCSs, e.g., distributed energy resources controlled by aggregators in a virtual power plant setting.

\bibliographystyle{IEEEtran}
\bibliography{Lit.bib}

\end{document}